\documentclass[
 reprint,
 amsmath,amssymb,
 aps,pra,
 floatfix,
 nofootinbib
]{revtex4-2}

\usepackage[T1]{fontenc}
\usepackage[utf8]{inputenc}
\usepackage{lmodern}
\usepackage{mathtools}
\usepackage{microtype}
\usepackage{amsthm}
\usepackage[hidelinks]{hyperref}

\hypersetup{
  pdftitle={Symmetry-guided constructions of absolutely maximally entangled states in five open cases},
  pdfauthor={Samuel Bevins; Yunus Bidav},
  pdfsubject={Absolutely maximally entangled states and quantum MDS codes},
  pdfkeywords={absolutely maximally entangled states, quantum MDS codes, Hermitian self-dual codes, superregular matrices}
}

\newcommand{\F}{\mathbb{F}}
\newcommand{\Z}{\mathbb{Z}}
\newcommand{\C}{\mathbb{C}}
\newcommand{\AME}{\operatorname{AME}}

\newcommand{\code}[3]{[\![#1,#2,#3]\!]}

\newtheorem{theorem}{Theorem}
\newtheorem{criterion}[theorem]{Criterion}

\begin{document}

\title{Symmetry-guided constructions of absolutely maximally entangled states in five open cases}

\author{Samuel Bevins}
\email{sjbevins@wm.edu}
\affiliation{Department of Physics, William \& Mary, Williamsburg, Virginia 23187, USA}

\author{Yunus Bidav}
\email{yunus@bidav.net}
\affiliation{Independent researcher}

\date{August 9, 2026}

\begin{abstract}
We give explicit Hermitian self-dual maximum distance separable codes
with parameters $[12,6,7]_{25}$, $[18,9,10]_{121}$, and
$[18,9,10]_{169}$.  The stabilizer construction proves the existence
of $\AME(12,5)$, $\AME(18,11)$, and $\AME(18,13)$ states; one-party
projection also gives $\AME(17,11)$ and $\AME(17,13)$.  The first code
was found by a direct search.  A regular $\Z_3^2$ coordinate orbit of
its automorphism group suggested a group-circulant form that reduces
each length-eighteen construction to a nine-entry kernel.  The printed
matrices are certified by exact Hermitian products and complete
square-minor enumeration.
\end{abstract}

\maketitle
\raggedbottom

\section{Result and strategy}
\label{sec:introduction}

A pure state in $(\C^q)^{\otimes n}$ is \emph{absolutely maximally
entangled}, denoted $\AME(n,q)$, when every subsystem of at most
$\lfloor n/2\rfloor$ parties is maximally mixed, meaning that its reduced
density matrix is proportional to the identity~\cite{Scott2004,Helwig2012}.
AME states connect quantum secret
sharing, combinatorial designs, multiunitary matrices, which remain
unitary under several index rearrangements, and quantum
error correction~\cite{Helwig2012,Goyeneche2015,Rajchel2026}.  Equivalently,
an $\AME(n,q)$ state is a pure (nondegenerate) quantum error-correcting
code $\code{n}{0}{\lfloor n/2\rfloor+1}_q$.  Its three entries give the
numbers of physical and encoded logical $q$-level systems and the
error-correcting distance.  For even $n$ this code saturates the
quantum Singleton bound, the fundamental length-distance limit for a
quantum code~\cite{Rains1999,HuberGrassl2020}.

Existence remains a parameter-by-parameter problem.  Weight enumerators
count codewords by their number of nonzero coordinates, while shadow
inequalities impose additional positivity constraints.  These coding
bounds give notable nonexistence results, including for four and seven
qubits~\cite{HiguchiSudbery2000,HuberSeven2017,HuberShadow2018}.
Across local
dimensions, seven-party AME states have recently been classified:
$\AME(7,q)$ exists exactly when $q\ge3$~\cite{ShiSeven2026}.  The
Huber--Wyderka AME table and Grassl's quantum-code tables organize much
of the known parameter landscape~\cite{HuberWyderka,GrasslTables}.

We establish the five existence statements
\begin{equation}
 \begin{aligned}
 &\AME(12,5),\quad \AME(17,11),\quad \AME(18,11),\\
 &\AME(17,13),\quad \AME(18,13).
 \end{aligned}
\label{eq:five-results}
\end{equation}
The proof has one essential ingredient: three systematic generator
matrices $G=[I_k\mid A]$, whose rows span the codes and whose first block
is the identity matrix.  We verify
\begin{equation}
 A\overline A^{\mathsf T}=-I_k
 \quad\text{and}\quad
 \det A[R,C]\ne0
 \label{eq:two-checks}
\end{equation}
for every pair of equally sized nonempty row and column sets $R,C$, with
$A[R,C]$ denoting the corresponding square submatrix.  These two checks
make the row space of $G$ both Hermitian self-dual and maximum distance
separable (MDS).  The standard stabilizer construction then gives the three
even-party AME states, and projection gives the two odd-party states.

The search that produced the matrices is secondary to this
certificate, but it explains their form.  A direct search at length
twelve produced the first code.  Its automorphisms revealed a regular
$\Z_3^2$ coordinate orbit.  Imposing the same translation symmetry on
two nine-coordinate orbits reduces an otherwise unrestricted
$9\times9$ block to nine field elements.  Character decomposition
then separates the self-duality equations before the MDS minors are
tested.  The remainder of the paper presents this reduction, gives the
three constructions, and records the exact verification.

\section{The coding criterion}
\label{sec:criterion}

Let $q$ be a prime power and write $\overline{x}=x^q$ for the
nontrivial Galois involution of $\F_{q^2}/\F_q$.  The Hermitian inner
product on $\F_{q^2}^{n}$ is
\begin{equation}
 \langle u,v\rangle_h=\sum_{j=1}^{n}u_j\overline{v_j}.
\end{equation}
The Hermitian dual $C^{\perp_h}$ consists of the vectors orthogonal to every
vector in $C$; Hermitian self-duality means $C=C^{\perp_h}$.  In the
notation $[n,k,d]$, the entries are length, dimension, and minimum
Hamming distance, the fewest coordinates in which two distinct codewords differ.
The code is MDS when $d=n-k+1$, attaining the classical Singleton bound.

\begin{criterion}[Systematic Hermitian MDS criterion]
\label{crit:systematic}
For $G=[I_k\mid A]$ over $\F_{q^2}$, the row space of $G$ is a
Hermitian self-dual $[2k,k,k+1]_{q^2}$ code if and only if both
conditions in Eq.~\eqref{eq:two-checks}
hold~\cite{MacWilliamsSloane1977,RothSeroussi1985}.
\end{criterion}

Indeed,
$G\overline G^{\mathsf T}=I_k+A\overline A^{\mathsf T}$, so the first
condition gives Hermitian self-duality.  Choosing any $k$ columns of
$G$ gives, up to sign, a square minor of $A$; hence all such column
sets are independent exactly when every square minor is nonzero.  This
is the systematic superregularity criterion for an MDS code; a square
minor is the determinant of a square submatrix.

The stabilizer correspondence maps an $\F_{q^2}$-linear Hermitian
self-dual MDS code $[2k,k,k+1]_{q^2}$ to a pure quantum MDS code
$\code{2k}{0}{k+1}_{q}$ and therefore an $\AME(2k,q)$ stabilizer
state, which is fixed by a commuting group of Pauli-type
operators~\cite{AshikhminKnill2001,Ketkar2006}.  In addition,
\begin{equation}
 \AME(2k,q)\ \Longrightarrow\ \AME(2k-1,q)
\label{eq:projection}
\end{equation}
by projecting any one party onto a basis state~\cite[Theorem~2]{Helwig2012}.

\section{Symmetry-reduced search}
\label{sec:search}

The $[12,6,7]_{25}$ code was found by a row-wise search that rejected
partial blocks when a Hermitian orthogonality condition failed or a
square minor vanished.  No coordinate symmetry was imposed.  A
subsequent automorphism calculation revealed a regular
$H\cong\Z_3^2$ orbit on nine coordinates, suggesting the
length-eighteen search form.  A coordinate automorphism is a code-preserving
permutation; regular means that exactly one element of $H$ maps any
chosen coordinate in the orbit to any other.
Restricting a code search by a prescribed automorphism group is an
established method~\cite{Camion1972,BraunKohnertWassermann2005,Wassermann2021};
the result is exhaustive only within that invariant class~\cite{Ostergard2021}.

Index both halves of $[I_9\mid A]$ by
$H=\Z_3\times\Z_3$.  Requiring simultaneous translation invariance
gives
\begin{equation}
 A_{x,y}=a(y-x),\qquad x,y\in H,
\label{eq:group-circulant}
\end{equation}
so the nine-element kernel $a:H\to\F_{q^2}$ determines all $81$ entries.
The block is group-circulant because each entry depends only on the group
difference $y-x$.  Hermitian
self-duality is equivalent to
\begin{equation}
\sum_{h\in H}a(h+g)\overline{a(h)}=-\delta_{g,0},
 \qquad g\in H.
\label{eq:convolution}
\end{equation}
Here $\delta_{g,0}$ is $1$ when $g=0$ and $0$ otherwise.

The character decomposition, which is the finite Fourier transform on
$H$, separates these equations~\cite{LingSole2001,JitmanLingSole2014,Palines2018}; related
Euclidean permutation-group analysis appears in Ref.~\cite{DeyRajan2004}.
Two regular coordinate orbits are required for Hermitian self-duality,
so length eighteen is the smallest positive length for this regular
action.

Enumerating self-dual subspaces in these character
blocks~\cite{RainsSloane1998} gives finite candidate families for
$q=11$ and $q=13$.  We
reconstructed candidates from their character blocks, screened minors
in increasing order, and stopped after finding the kernels below.
Thus the searches identify valid constructions but do not classify the
invariant families.

\section{Explicit constructions and proof}
\label{sec:constructions}

For $q=5$, let
\begin{equation}
 \F_{25}=\F_5(\alpha),\qquad \alpha^2+\alpha+2=0,
\end{equation}
and encode $a+5b$ as $a+b\alpha$, with $0\le a,b<5$.  Define
$G_5=[I_6\mid A_5]$, where
\begin{equation}
A_5=\begin{pmatrix}
1&1&1&1&1&2\\
23&1&2&3&5&3\\
19&7&1&22&17&12\\
9&15&11&2&18&15\\
1&4&22&9&2&1\\
6&7&18&15&10&2
\end{pmatrix}.
\label{eq:A5}
\end{equation}
For a linear code $C\subseteq\F^n$, its Schur square is the span of all
coordinatewise products $x\star y=(x_1y_1,\ldots,x_ny_n)$ with $x,y\in C$.
A generalized Reed--Solomon (GRS) $[n,k]$ code consists of evaluations
of polynomials of degree less than $k$ at $n$ distinct field points, with a
fixed nonzero scale factor in each coordinate.  Codes are monomially
equivalent if coordinate permutation and nonzero coordinate scaling
convert one into the other; these operations preserve Schur-square
dimension.  Here $\dim C^{\star2}=12$, while every GRS $[12,6]$ code has
$\dim C^{\star2}\le11$, so the constructed code is not monomially
equivalent to a GRS code~\cite{Cascudo2015}.

For $q\in\{11,13\}$, order the elements of $H$ as
\begin{equation}
(0,0),(1,0),(2,0),(0,1),(1,1),(2,1),(0,2),(1,2),(2,2).
\label{eq:H-order}
\end{equation}
Read each matrix below as the values $a_q(u,v)$, with row index $u$ and
column index $v$.  Thus Eq.~\eqref{eq:H-order} reads the kernel columns
from left to right.  Reconstruct $A_q$ using
Eq.~\eqref{eq:group-circulant}.

For $q=11$, use
\begin{equation}
 \F_{121}=\F_{11}(\gamma),\qquad \gamma^2=-1,
\end{equation}
and the kernel
\begin{equation}
\bigl(a_{11}(u,v)\bigr)=
\begin{pmatrix}
4+8\gamma&10+3\gamma&9+10\gamma\\
8+10\gamma&8\gamma&8+10\gamma\\
9+10\gamma&10+3\gamma&4+8\gamma
\end{pmatrix}.
\label{eq:kernel-11}
\end{equation}
Conjugation sends $a+b\gamma$ to $a-b\gamma$.

For $q=13$, use
\begin{equation}
 \F_{169}=\F_{13}(\beta),\qquad \beta^2+\beta+2=0,
\end{equation}
and the kernel
\begin{equation}
\bigl(a_{13}(u,v)\bigr)=
\begin{pmatrix}
6+12\beta&9+9\beta&1+7\beta\\
1+10\beta&11+7\beta&5\beta\\
8+12\beta&6+9\beta&2+7\beta
\end{pmatrix}.
\label{eq:kernel-13}
\end{equation}
Here conjugation sends $c_0+c_1\beta$ to
$(c_0-c_1)-c_1\beta$.

\begin{theorem}
\label{thm:main}
The matrices in Eqs.~\eqref{eq:A5},~\eqref{eq:kernel-11},
and~\eqref{eq:kernel-13} define Hermitian self-dual MDS codes
\begin{equation}
[12,6,7]_{25},\qquad [18,9,10]_{121},\qquad [18,9,10]_{169}.
\end{equation}
Consequently all five AME states in Eq.~\eqref{eq:five-results} exist.
\end{theorem}

\begin{proof}
Exact arithmetic in the stated polynomial bases gives
\begin{equation}
A_5\overline A_5^{\mathsf T}=-I_6,\qquad
A_q\overline A_q^{\mathsf T}=-I_9\quad(q=11,13).
\end{equation}
For the length-eighteen matrices, these identities are equivalent to
the nine equations in Eq.~\eqref{eq:convolution}.  Gaussian elimination
over the same fields shows that every nonempty square minor is nonzero.
The number of nonempty square minors of a $k\times k$ matrix is
\begin{equation}
\sum_{s=1}^{k}\binom{k}{s}^2=\binom{2k}{k}-1,
\end{equation}
so we checked $923$ minors for $A_5$ and $48{,}619$ for each $A_q$.
Criterion~\ref{crit:systematic} proves the three classical-code
statements.  The Hermitian stabilizer construction proves
$\AME(12,5)$, $\AME(18,11)$, and $\AME(18,13)$, while
Eq.~\eqref{eq:projection} gives $\AME(17,11)$ and $\AME(17,13)$.
\end{proof}

\section{Context and conclusion}
\label{sec:conclusion}

The $q=5$ and $q=11$ parameters in Eq.~\eqref{eq:five-results} were
listed as open in the Huber--Wyderka existence table~\cite{HuberWyderka}.
That table does not extend to
$q=13$; the two $q=13$ parameters were unresolved in the sources
reviewed here.

Hermitian self-dual GRS codes are confined to even lengths
$n\le q+1$~\cite{KimLee2004,TongWang2017,ZhaoMa2026}, so length eighteen lies
outside the GRS class for both $q=11$ and $q=13$.  Previously reported
Hermitian self-dual codes at these length-eighteen parameters had
distance nine rather than the MDS distance ten~\cite{SokYang2022}.
Although an earlier extended-duadic theorem would nominally include the
$q=13$ case~\cite{Guenda2012}, its self-duality hypothesis was later
corrected, and the required duadic splitting, a partition of the
relevant residue classes into two multiplier-related halves, fails
here~\cite{TongWang2017}.  At $q=5$ and length twelve, published self-dual
constructions likewise reached distance six rather than
seven~\cite{GrasslRoetteler2015,Suprijanto2014,SokYang2022}.

Quantum MDS
codes with the same length and distance were known at local dimensions
five and seven, but do not provide the $q=11$ or $q=13$ states proved
here~\cite{Ball2021}.

The present theorem proves existence but does not classify the states.
Local-unitary equivalence allows each party to act unitarily on its own
subsystem; stochastic local operations and classical communication
(SLOCC) equivalence also permits nonunitary operations with nonzero success
probability.  Related work provides combinatorial
constructions, inequivalent families, and classifications of smaller
AME cases~\cite{Rather2022,Rather2023,BurchardtRaissi2020,
Ramadas2025,Tan2026,Rajchel2026}.

For prime local dimensions, stabilizer states also admit graph-state
representatives~\cite{Schlingemann2002,BahramgiriBeigi2007}.  In this
description an adjacency matrix encodes the state, and the AME property
becomes a full-rank condition on the block joining the two sides of each
relevant bipartition~\cite{HelwigGraph2013}.

In short, a direct search revealed a $\Z_3^2$ orbit, which led to
group-circulant kernels and exact MDS certificates.  The symmetry
reduces the search, but it is not part of the logical
certificate: the three printed matrices and the two checks in
Eq.~\eqref{eq:two-checks} suffice to prove the theorem.  The searches
were not exhaustive, so classification and equivalence questions
remain open.

\section*{Author contributions and AI use}

S.B. and Y.B. contributed equally.  Under the authors' direction,
Claude Fable 5 (Anthropic) and ChatGPT 5.6 Sol (OpenAI) were used
extensively throughout the project, including the entire computational
search, the development and implementation of search methods, exact
verification, bibliographic checks, and manuscript preparation.  Both
authors independently reviewed the constructions, computations, and
full text.

\begin{acknowledgments}
The authors are especially grateful to Felix Huber for his highly
impactful feedback, including his incisive questions, which substantially
improved this manuscript.  We thank Felix
Huber and Nikolai Wyderka for creating and maintaining the table of
absolutely maximally entangled states.  We also thank the maintainers of
the public quantum-code tables for making parameter-level comparisons
possible.
\end{acknowledgments}

\interlinepenalty=10000
\bibliography{references}
\end{document}